\documentclass[draftclsnofoot,onecolumn]{IEEEtran}

\usepackage{amssymb,amsfonts,amsmath,amstext,mathrsfs}
\usepackage{latexsym}
\usepackage{epsfig,psfrag,color,graphics,epsf}
\usepackage{enumerate,cite}

\usepackage{amsmath}
\usepackage{amssymb}
\usepackage{amsfonts}
\usepackage{amsthm}
\usepackage{graphicx}
\usepackage{float}
\usepackage{graphics}
\usepackage{graphicx}
\usepackage{euscript}
\usepackage{psfrag}

\newtheorem{theorem}{Theorem}
\newtheorem{lemma}{Lemma}

\newtheorem{corollary}{Corollary}
\newtheorem{definition}{Definition}

\newcommand {\bq} {\mbox{\boldmath $q$}}

\newcommand {\bx} {\mbox{\boldmath $x$}}
\newcommand {\by} {\mbox{\boldmath $y$}}

\newcommand {\bW} {\mbox{\boldmath $W$}}

\newcommand{\calA}{{\cal A}}

\newcommand{\calC}{{\cal C}}

\newcommand{\calE}{{\cal E}}

\newcommand{\calP}{{\cal P}}

\newcommand{\calX}{{\cal X}}
\newcommand{\calY}{{\cal Y}}

\begin{document}

\sloppy

\title{Converse Theorems for the DMC with Mismatched Decoding} 

\author{
  Anelia Somekh-Baruch\thanks{A.\ Somekh-Baruch is with the Faculty of Engineering at Bar-Ilan University, Ramat-Gan, Israel.  Email: somekha@biu.ac.il. 
  This work was supported by the Israel Science Foundation (ISF) under grant 2013/919. 
  Some of the results of this paper were presented at the IEEE International Symposium on Information Theory (ISIT) 2015 and at the 2016 International Zurich Seminar on Communications. This paper was accepted for publication in the IEEE Transactions on Information Theory.}
}
\maketitle

\begin{abstract}
The problem of mismatched decoding with an additive metric $q$ for a discrete memoryless channel $W$ is addressed. The ``product-space" improvement of the random coding lower bound on the mismatch capacity, $C_q^{(\infty)}(W)$, was introduced by Csisz\'{a}r and Narayan. It is equal to the supremum, as $k$ tends to infinity, of the rates that are achievable by random coding for the product channel $W^k$. 

We study two kinds of decoders. The {\it $\delta$-margin mismatched decoder} outputs a message whose metric with the channel output exceeds that of all the other codewords by at least $\delta$. The {\it $\tau$-threshold decoder} outputs a single message whose metric with the channel output exceeds a threshold $\tau$. Both decoders declare an error if they fail to find a message that meets the requirement. 
It is assumed that $q$ is bounded. 

It is proved that $C_q^{(\infty)}(W)$ is equal to the mismatch capacity with a constant margin decoder.  
We next consider sequences of $P$-constant composition codebooks, whose empirical distribution of the codewords are at least $o(n^{-1/2})$ close in the $L_1$ distance sense to $P$. Using the Central Limit Theorem, it is shown that for such sequences of codebooks the supremum of achievable rates with constant threshold decoding  is upper bounded by the supremum of the achievable rates with a constant margin decoder, and therefore also by $C_q^{(\infty)}(W)$. 

Further, a soft converse is proved stating that if the average probability of error of a sequence of codebooks converges to zero sufficiently fast, the rate of the code sequence is upper bounded by $C_q^{(\infty)}(W)$. In particular, if $q$ is a bounded rational metric, and the average probability of error converges to zero faster than $O(n^{-1})$, then $R\leq C_q^{(\infty)}(W)$. Finally, a max-min multi-letter upper bound on the mismatch capacity that bears some resemblance to $C_q^{(\infty)}(W)$ is presented.

\end{abstract}
\vspace{10cm}

\pagebreak

\section{Introduction}

Mismatched decoding is a prevalent channel coding paradigm in which the decoder has a fixed structure that cannot be tailored to the actual channel in use. 
The fixed structure of the decoder may be due to practical considerations such as decoding complexity or inaccurate channel estimation. 
It is usually assumed that among all codewords, the decoding rule maximizes a certain metric between the channel output sequence and the codeword. The encoder, who knows the metric used by the decoder, needs to construct an encoding scheme that will yield a vanishingly small average probability of error. The highest achievable rate using a given decoder metric is referred to as the {\it mismatch capacity}. 
 
 Mismatched decoding has been studied extensively, especially for discrete memoryless channels (DMCs). The formula for the achievable rates for the DMC with mismatched decoding using random coding, which is referred to as the LM rate, was derived by Csisz\'{a}r and K{\"o}rner \cite{CsiszarKorner81graph}
and by Hui \cite{Hui83}. In \cite{CsiszarKorner81graph}, 
an error exponent for random coding with fixed composition codes and mismatched decoding was established using a graph decomposition theorem. Lapidoth \cite{Lapidoth96} introduced an improved lower bound on the mismatch capacity of the DMC by studying the achievable sum-rate of an appropriately chosen multiple access channel (MAC) with mismatched decoding, whose codebook was obtained by expurgating codewords from the product of the codebooks of the two users. In \cite{SomekhBaruch_mismatchachievableIT2014},\cite{SomekhBaruchISIT_2013} the achievable region and error exponents of a {\it cognitive} MAC were considered using superposition coding or random binning whose sum-rate serves as a lower bound on the capacity of the single-user channel. An improved bound was presented by Scarlett et al.\ \cite{ScarlettMartinezGuilleniFabregasISIT_2013} (see also \cite{ScarlettAlfonsoFabregasArXivNOV2013}) using a refinement of the superposition coding ensemble. For given auxiliary random variables, the results of \cite{SomekhBaruch_mismatchachievableIT2014,SomekhBaruchISIT_2013,ScarlettMartinezGuilleniFabregasISIT_2013} may yield improvement of the achievable rates of \cite{Lapidoth96} for the DMC.
For other related works and extensions see 
\cite{Balakirsky_conference_95,ShamaiKaplan1993information,MerhavKaplanLapidothShamai94,Lapidoth96b,GantiLapidothTelatar2000,ShamaiSason2002,ScarlettFabregas2012,ScarlettAlfonsoFabregasArXivNOV2013,ScarlettMartinezGuilleniFabregas_mismatch_2014_IT,ScarlettMartinezGuilleniFabregas2012AllertonSU,ScarlettPengMerhavMartinezGuilleniFabregas_mismatch_2014_IT} and references therein.

In \cite{CsiszarNarayan95}, the mismatch capacity of the DMC with a decoding metric $q$, denoted $C_q(W)$, was considered. It was shown that the LM rate is not tight in general but that its positivity is a necessary condition for positive mismatch capacity. This result was obtained by proving that the random coding bound for the product channel $W_{Y_1,...,Y_k|X_1,...,X_k}=\prod_{i=1}^k W_{Y_i|X_i}$ ($k$ consecutive channel uses of the DMC $W$), denoted $C_q^{(k)}(W)$, could result in strictly higher achievable rates. They referred to the improved bound as the "product-space" improvement of the lower bound, and the supremum of the achievable rates obtained by taking the limit of $C_q^{(k)}(W)$ as $k$ tends to infinity was denoted $C_q^{(\infty)}(W)$. 

In the special case of erasures-only (e.o.) capacity, the product space improvement $C^{(\infty)}_q(W)$ was shown to be tight. Csisz\'{a}r and Narayan conjectured that the mismatch capacity is indeed given by $C^{(\infty)}_q(W)$. Whether this conjecture is true in general remains an open question.

It was further stated in \cite{CsiszarNarayan95} that ``although the bound is not computable, its tightness would afford some valuable conclusions, for instance, that for $R<C_q(W)$, codes with $d$-decoding always exist with rates approaching $R$ and probability of error approaching zero exponentially fast". Another implication of an affirmative answer to the conjecture concerns the threshold capacity of the DMC. The threshold capacity is the supremum of achievable rates obtained by decoding the unique message which accumulates a metric that exceeds a predetermined threshold. It was stated in \cite{CsiszarNarayan95} that if the conjecture is true, the threshold capacity and the mismatch capacity of the DMC are equal.

Unlike lower bounds, upper bounds on the mismatch capacity have only been provided in some special cases. 
The only non-trivial single-letter converse result reported in \cite{Balakirsky95} for binary-input DMCs was recently disproved in \cite{ScarlettSomehkBaruchMartinezGuilleniFabregas2015}. Specifically, a rate based on superposition coding was shown to exceed the claimed mismatch capacity of \cite{Balakirsky95}.

In \cite{SomekhBaruch_general_formula_IT2015} a general formula was derived in the form of a Verd\'{u}-Han \cite{VerduHan1994} expression for the mismatch capacity of a general channel, and was defined as a sequence of conditional distributions with a general decoding metric sequence. 
Further, several upper bounds on the capacity were provided, and a simpler expression for a lower bound was derived for the case of a non-negative decoding metric. The general formula was specialized to the case of finite input and output alphabet channels with a type-dependent metric. 
The problem of threshold mismatched decoding was also studied, and a general expression for the threshold mismatch capacity was obtained.

In this paper we focus on a DMC with mismatched decoding with an additive metric, and study several properties of the ``product-space" improvement of the random coding lower bound on the mismatch capacity, $C_q^{(\infty)}(W)$.
We define a {\it mismatched decoder with a $\delta$-margin} as a decoder that decides in favor of the message whose metric with the channel output exceeds that of all the other codewords by at least $\delta$: if no such prominent codeword exists, an error is declared.  
The mismatch capacity with a $\delta$-margin is the supremum of achievable rates with mismatched $\delta$-margin decoders. 
The mismatch capacity with a constant margin decoder is the supremum over $\delta>0$ of the mismatch capacity with a $\delta$-margin. 
The word {\it constant} is used to emphasize the fact that the margin level $\delta$ is kept constant for all block lengths. 
Note that maximum likelihood (ML) decoding with a $\delta_n$-margin was studied by Lapidoth and Ziv \cite[Equation (17)]{LapidothZiv98}, where it was called a threshold decoder (see also \cite{Hashimoto1999}). Nevertheless, since threshold decoding usually refers to a decoder that outputs the single codeword whose metric exceeds a certain threshold, we chose to adopt the term ``margin decoder" as stated above to distinguish between the two terms. 

The first result presented in this paper is a proof that $C_q^{(\infty)}(W)$ is equal to the mismatch capacity with a constant margin decoder.  
The significance of the result stems from the fact that if we can prove that the capacity with a constant margin decoder of the DMC is equal to the mismatch capacity, it could be inferred that the mismatch capacity is equal to $C_q^{(\infty)}(W)$.
We also introduce the notion of constant threshold decoding, which is the supremum of achievable rates with threshold decoding with a threshold level that is kept constant for all block lengths. 

We establish an inequality between the supremum of achievable rates using constant margin decoding and constant threshold decoding, under the assumption that the sequence of codebooks is $P$-constant composition (a term that will be defined rigorously).
 We further prove a soft converse stating that if the minimum probability of error achievable for block-length $n$ converges to zero sufficiently fast, the rate of the code is upper bounded by $C_q^{(\infty)}(W)$.
In particular, if $q$ is a bounded rational metric, and the lowest achievable average probability of error, $\epsilon_n$, satisfies $\lim_{n\rightarrow\infty}n\epsilon_n=0$, then $R\leq C_q^{(\infty)}(W)$. 
Since for any $k$, at rates below $C_q^{(k)}(W)$ there exist code sequences with an average probability of error that vanishes exponentially fast, 
this implies that $C_q^{(\infty)}(W)$ is the supremum of the rates below which there exist codes with an average probability of error that vanishes exponentially fast. 

We further present an upper bound on the transmission rate: a max-min multi-letter upper bound on the mismatch capacity $C_q(W)$ is derived. 
This upper bound can be thought of as the product space improvement of the random coding {\it lower bound} with a genie-aided decoder that is informed of the joint empirical statistics of the output symbols $y_i, i=1...,n$ and their metrics with respect to the channel inputs, $q(x_i,y_i), i=1,...,n$. 
Some of the results of this paper were presented in \cite{SomekhBaruchISIT2015MultiletterMismatchedDMC}.

The remainder of this paper is as follows. In Section \ref{sc: Notation} we present some notational conventions. 
Section \ref{sc:  Preliminaries} provides a formal statement of the communication channel setup and definitions.  
In Section \ref{sc: Main Results}, 
an equality between $C_q^{(\infty)}(W)$ and the 
capacity with a constant margin decoder is established, a relationship between the achievable rate of a constant margin decoder and a constant threshold decoder is introduced, and $C_q^{(\infty)}(W)$ is characterized as the highest achievable average probability of error which decays sufficiently fast. Section \ref{sc: Max-Min Upper Bounds} is devoted to deriving a max-min upper bound on the mismatch capacity of the DMC.  
The discussion and concluding remarks appear in Section  \ref{sc: Conclusion}.

\section{Notation}\label{sc: Notation}

Throughout this paper, scalar random variables are denoted by capital letters, their sample values are denoted by their respective lower case letters, and their alphabets are denoted by their respective calligraphic letters, e.g.\ $X$, $x$, and $\calX$, respectively. A similar convention applies to random vectors of dimension $n$ and their sample values, which are 
superscripted by $n$, e.g.., $x^n$. The set of all $n$-vectors with components taking values in a certain finite alphabet are denoted by the same alphabet superscripted by n, e.g., $\calX^n$. 
Logarithms are taken to the natural base $e$. 

For a given sequence $\by \in \calY^n$, where $\calY$ is a finite alphabet,  $\hat{P}_{\by}$ denotes the empirical distribution on $\calY$ extracted from $\by$; in other words, $\hat{P}_{\by}$ is the vector $\{ \hat{P}_{\by} (y), y\in\calY\}$, where $ \hat{P}_{\by} (y)$ is the relative frequency of the symbol $y$ in the vector $\by$. The type-class of $\bx$ is the set of $\bx'\in\calX^n$ such that $\hat{P}_{\bx'}=\hat{P}_{\bx}$, which is denoted $T(\hat{P}_{\bx})$. 
Let $\calP(\calX)$ denote the set of all probability distributions on $\calX$, The set of empirical distributions of order $n$ on alphabet $\calX$ is denoted  $\calP_n(\calX)$.

Information theoretic quantities such as entropy, conditional entropy, and mutual information are denoted following the usual conventions in the information theory literature, e.g., $H (X )$, $H (X |Y )$, $I(X;Y)$ and so on. To emphasize the dependence of the quantity on a certain underlying probability distribution, say $\mu$, it is subscripted by $\mu$, e.g., with notations such as $H_\mu(X )$, $H_\mu(X |Y)$, $I_\mu(X;Y)$, etc. The expectation operator is denoted by $\mathbb{E} (\cdot)$, and once again, to make the dependence on the underlying distribution $\mu$ clear, it is denoted by $\mathbb{E}_\mu(\cdot)$. The cardinality of a finite set $\calA$ is denoted by $|\calA|$. The indicator function of an event $\calE$ is denoted by $1\{\calE \}$.

\section{Preliminaries}\label{sc: Preliminaries}

Consider a DMC with a finite input alphabet 
$\calX$ and a finite output alphabet $\calY$, which is governed by the conditional p.m.f.\ $W$. 
As the channel is fed by an input vector $x^n \in\calX^n$, it generates an output vector $y^n \in\calY^n$ according to the sequence of conditional probability distributions 
\begin{equation}P (y_i |x_1 , . . . , x_i , y_1 , . . . , y_{i-1} ) = W(y_i|x_i), \quad i = 1, 2, . . . , n\end{equation}
where for $i = 1, (y_1 , . . . , y_{i-1}) $ is understood as the null string.

A rate-$R$ block-code of length $n$ consists of $ e^{nR}$  $n$-vectors $x^n(m)$, $m = 1, 2, . . . , e^{nR}$, which represent $e^{nR}$ different messages; i.e., it is defined by the encoding function
\begin{flalign}
f_n:\; \{1,...,e^{nR}\} \rightarrow \calX^n.
\end{flalign}
It is assumed that all the possible messages are a-priori equiprobable; i.e., $P (m) = e^{-nR}$ for all $m\in\left\{1,...,e^{nR}\right\}$.  
A mismatched decoder for the channel is defined by a mapping 
\begin{flalign}\label{eq: qn mapping}
q_n:\;  \calX^n\times \calY^n\rightarrow  \mathbb{R},
\end{flalign}
where the decoder declares that message $i$ was transmitted iff 
\begin{flalign}\label{eq: decoder decision rule}
q_n(x^n(i),y^n)>q_n(x^n(j),y^n), \forall j\neq i,
\end{flalign}
and if no such $i$ exists, an error is declared. 

An important class of mismatched decoders is the class of additive decoding functions; i.e.,  
\begin{flalign}\label{eq: additive decoder decision rule}
 q_n(x^n,y^n)=\frac{1}{n}\sum_{i=1}^n q(x_i,y_i), 
\end{flalign}
where $q$ is a mapping from $\calX\times \calY$ to $\mathbb{R}$. 

The results that will be presented in this paper hold under one of the following boundedness assumptions (to be specified when needed). 

The first assumption is that $q$ is bounded as follows: there exists $B\geq 0$ such that 
\begin{flalign}\label{eq: bounded q  aaa}
|q(x,y)|\leq B<\infty,\; \forall (x,y)\in\calX\times\calY:\; W(y|x)>0.
\end{flalign}
The second assumption is slightly more restrictive: there exists $B\geq 0 $ such that 
\begin{flalign}\label{eq: boundedness updated}
q(x,y)\leq B<\infty,\;  \forall (x,y)\in\calX\times\calY \mbox{ and }  q(x,y) \geq -B>-\infty,\; \forall (x,y)\in\calX\times\calY :\; W(y|x)>0.
\end{flalign}
Note that the matched metric $q(x,y)=\log W(y|x)$ satisfies both assumptions (\ref{eq: bounded q  aaa}) and (\ref{eq: boundedness updated}).

Below are several useful definitions which refer to a general channel $W^{(n)}$ from $\calX^n$ to $\calY^n$ and to the DMC $W^{(n)}=W^n$ as a special case.
It is useful to define the average probability of error associated with a codebook, a channel and a metric:
\begin{definition}\label{eq: P_e W calC_n q_n dfn}
For a given codebook $\calC_n$, let $P_e(W^{(n)},\calC_n,q_n)$ designate the average probability of error incurred by the decoder $q_n$ (\ref{eq: decoder decision rule}) employed on the output of the channel $W^{(n)}$. 
\end{definition}
\begin{definition}\label{eq:  nMesiloncode dfn}
A code $\calC_n$ with a decoding metric $q_n$ is an $\left(n,M,\epsilon\right)$-code for the channel $W^{(n)}$ if 
it has $M$ codewords of length $n$ and $P_e(W^{(n)},\calC_n,q_n) \leq
\epsilon$. 
\end{definition}
We next define an $\epsilon$-achievable rate and the mismatch capacity. 
\begin{definition}\label{dfn: rpsilon achievable}
A rate $R>0$ is an $\epsilon$-achievable rate for the channel $\bW=\left\{W^{(n)}\right\}_{n\in \mathbb{N}}$ with decoding metric sequence $\bq=\{q_n\}_{n\in \mathbb{N}}$ if for every $\delta>0$, there exists a sequence of codes $\{\calC_n\}_{n\in \mathbb{N}}$ such that for all $n$ sufficiently large, $\calC_n$ is an $\left(n,M_n,\epsilon\right)$ code for the channel $W^{(n)}$ and decoding metric $q_n$ with rate $\frac{\log(M_n)}{n}\geq R-\delta$.  
\end{definition}
\begin{definition}
The mismatch capacity of channel $W$ with an additive decoding metric $q$, denoted $C_q(W)$, is the supremum of rates that are $\epsilon$-achievable for all $0<\epsilon<1$. \end{definition}

As mentioned in the introduction, it was proved in \cite{CsiszarKorner81graph}, \cite{Hui83} that the mismatch capacity of the DMC is lower bounded by the LM rate given by the formula
\begin{flalign}\label{eq: Cq1 dfn}
C_q^{(1)}(W)=& \max_{P_X} \underset{P_{\widetilde{Y}|X}:\; P_{\widetilde{Y}}=P_{Y}, \mathbb{E}(q(X,\widetilde{Y}))\geq \mathbb{E}(q(X,Y))  }{\min}I(X;\widetilde{Y}),
\end{flalign}
where $P_X\in\calP(\calX)$, $(X,Y)\sim P_X\times W$, $(X,\widetilde{Y})\sim P_X\times P_{\widetilde{Y}|X}$, and $P_Y,P_{\widetilde{Y}}$ are the corresponding marginal distributions of $Y,\widetilde{Y}$, respectively. By considering the achievable rate for channel $W^k$ from $\calX^k$ to $\calY^k$, the following rate is also achievable (\hspace{1sp}\cite{CsiszarNarayan95})
\begin{flalign}\label{eq: Cqk dfn}
C_q^{(k)}(W)=& \max_{P_{X^k}} \underset{P_{\widetilde{Y}^k|X^k}:\; P_{\widetilde{Y}^k}=P_{Y^k}, \mathbb{E}(q_k(X^k,\widetilde{Y}^k))\geq \mathbb{E}(q_k(X^k,Y^k))  }{\min}\frac{1}{k}I(X^k;\widetilde{Y}^k),
\end{flalign}
where $P_{X^k}\in\calP(\calX^k)$, $(X^k,Y^k)\sim P_{X^k}\times W^k$, $(X^k,\widetilde{Y}^k)\sim P_{X^k}\times P_{\widetilde{Y}^k|X^k}$, and $P_{Y^k},P_{\widetilde{Y}^k}$ are the corresponding marginal distributions of $Y^k,\widetilde{Y}^k$, respectively. Since for all $k$ the rate $C_q^{(k)}(W)$ is achievable, then so is
\begin{flalign}\label{eq: C q infty dfn}
C_q^{(\infty)}(W)=& \mbox{limsup}_{k\rightarrow\infty}C_q^{(k)}(W).
\end{flalign}

A closely related notion to that of a mismatched decoder is the $q_n$-decoder with a $\delta$-margin, which decides that $i$ is the transmitted message iff
\begin{flalign}
q_n(x^n(i),y^n)\geq q_n(x^n(j),y^n)+\delta , \forall j\neq i,\label{eq: margin decision rule 1}
\end{flalign}
and if no such $i$ exists an error is declared.

The following definition extends $P_e(W^{(n)},\calC_n,q_n)$ to the case of a margin decoder. 
\begin{definition}\label{eq: P_e W calC_n q_n tau_n dfn}
For a given codebook $\calC_n$, let $P_{margin}(W^{(n)},\calC_n,q_n,\delta)$ designate the average probability of error incurred by the $q_n$ decoder with a $\delta$-margin (\ref{eq: margin decision rule 1}) employed on the output of the channel $W^{(n)}$.
\end{definition}
An $\left(n,M,\epsilon\right)$-code and an $\epsilon$-achievable rate with respect to $\delta$-margin decoding are defined similarly to Definitions \ref{eq:  nMesiloncode dfn} and \ref{dfn: rpsilon achievable}. 
The mismatch capacity using the additive decoding metric $q_n$ (\ref{eq: additive decoder decision rule}) with a $\delta$-margin is denoted $C_{q,margin}(W,\delta)$. 
Finally, we define the mismatch capacity with a constant margin decoder as 
\begin{flalign}
C_{q,margin}^{const}(W)\triangleq \sup_{\delta>0}C_{q,margin}(W,\delta).
\end{flalign}

It is easy to see that $C_{q,margin}^{const}(W)\leq C_q(W)$ since the margin decoder is more restrictive and if the $q_n$ decoder errs, the $q_n$ decoder with a $\delta$-margin errs as well.

Another decoding rule that is related to mismatched $q_n$-decoder is the $(q_n,\tau)$-threshold decoder which decides that $i$ is the transmitted message iff
\begin{flalign}
q_n(x^n(i),y^n)\geq \tau \label{eq: threshold decision rule 1}
\end{flalign}
and
\begin{flalign}
q_n(x^n(j),y^n)<\tau,\;\forall j\neq i.\label{eq: threshold decision rule 2}
\end{flalign}
The following definition extends $P_e(W^{(n)},\calC_n,q_n)$ to the case of a threshold decoder, 
\begin{definition}\label{eq: P_e W calC_n q_n tau_n dfn}
For a given codebook $\calC_n$, let $P_{thresh}(W^{(n)},\calC_n,(q_n,\tau))$ designate the average probability of error incurred by the $(q_n,\tau)$-threshold decoder (\ref{eq: threshold decision rule 1})-(\ref{eq: threshold decision rule 2}) employed on the output of channel $W^{(n)}$.
\end{definition}
An $\left(n,M,\epsilon\right)$-code and an $\epsilon$-achievable rate with respect to $\tau$-threshold decoding are defined similarly to Definitions \ref{eq:  nMesiloncode dfn} and \ref{dfn: rpsilon achievable}. 
\begin{definition}\label{df: constant threshold capacity}
The constant threshold $q$-capacity of a channel denoted $C_{q,threshold}^{const}(W)$ is defined as the supremum over $\tau$ of the rates achievable by
codes with $(q_n,\tau)$-threshold decoders. 
\end{definition}
One has 
$
 C_{q,threshold}^{const}(W)\leq C_q(W)$,
since a threshold decoder (\ref{eq: threshold decision rule 1})-(\ref{eq: threshold decision rule 2}) is more restrictive than the mismatched decoder (\ref{eq: decoder decision rule}).

\section{Properties of $C_q^{(\infty)}(W)$}\label{sc: Main Results}

In this section we prove several properties of $C_q^{(\infty)}(W)$. The first property is that $C_q^{(\infty)}(W)$ is equal to $C_{q,margin}^{const}(W)$ and the second is that $C_q^{(\infty)}(W)$ is the highest achievable rate for an average probability of error $P_e(W^{(n)},\calC_n,q_n)$  which vanishes sufficiently fast. 
Further, we consider sequences of $P$-constant composition codebooks, such that the empirical distribution of the codewords converges to $P$ (at least $o(n^{-1/2})$-fast). It is shown 
that the supremum of the achievable rates with constant threshold decoding for such sequences of codebooks is upper bounded by the supremum of the achievable rates with a constant margin decoder, and therefore also by $C_q^{(\infty)}(W)$. 

Before proving these properties we recall that 
the dual expression for the $n$-letter product-form expression of the LM rate (\hspace{1sp}\cite{MerhavKaplanLapidothShamai94},\cite{GantiLapidothTelatar2000}, \cite{ShamaiSason2002}) in the case of an additive metric can be expressed as
\begin{flalign}\label{eq: DUAL EXPRESSION MODIFIED}
C_q^{(n)}(W)=& \frac{1}{n}\max_{P^{(n)}\in\calP(\calX^n)}\sup_{s_n\geq 0, a(x^n)\geq 0} \mathbb{E}\left(\log\frac{e^{s_nq_n(X^n,Y^n)+a(X^n)}}{\sum_{\bar{x}^n}{P}^{(n)}(\bar{x}^n)e^{s_nq_n(\bar{x}^n,Y^n)+a(\bar{x}^n)}}\right),
\end{flalign}
where $(X^n,Y^n)\sim P^{(n)}\times W^n$. 
The following lemma will be useful in lower bounding $C_q^{(n)}(W)$ and therefore also $C_q^{(\infty)}(W)$.
\begin{lemma}\label{lm: essential lemma}
Let $q$ be bounded as in \eqref{eq: boundedness updated}, let $\calA\subseteq \calX^n\times\calY^n$, fix $s_n\geq 0$, and fix a subset $\calC_n\subseteq \calX^n$ of size $|\calC_n|=e^{nR_n}$. Let $P^{(n)}$ denote the uniform distribution over $\calC_n$, and let $(X^n,Y^n)\sim P^{(n)}\times W^n$, then
\begin{flalign}\label{eq: main lemma result}
&\sup_{ a(x^n)\geq 0}\mathbb{E}\left(\log\frac{e^{s_nq_n(X^n,Y^n)+a(X^n)}}{\sum_{\bar{x}^n}{P}^{(n)}(\bar{x}^n)e^{s_nq_n(\bar{x}^n,Y^n)+a(\bar{x}^n)}}\right)\nonumber\\
&\geq - \Pr\left\{\calA\right\}\cdot \frac{s_n\cdot 2B}{n}  \nonumber\\
&\quad + \Pr\left\{\calA^c\right\}\cdot  \mathbb{E}\left( \min\left\{R_n, \frac{s_n}{n}\cdot \min_{\bar{x}^n\in\calC_n:\; \bar{x}^n\neq X^n }\left[ q_n(X^n,Y^n)-q_n(\bar{x}^n,Y^n)\right]  \right\}\bigg |\calA^c\right)-\frac{1}{n}\log(2),
\end{flalign}
where $\Pr\left\{\calA\right\}=\Pr\left\{(X^n,Y^n)\in \calA\right\}$, and $B$ is the metric bound in \eqref{eq: boundedness updated}.
\end{lemma}
 Lemma \ref{lm: essential lemma} will be used in the proofs of Theorems \ref{th: Cinfinity equals C margin} and \ref{th: rational theorem}, with different substitutions for the event $\calA$. 
 In Theorem \ref{th: Cinfinity equals C margin}, $\calA^c$ will stand for the event of $\delta$-margin correct decoding, under which the term $\min_{\bar{x}^n\in\calC_n:\; \bar{x}^n\neq X^n }\left[ q_n(X^n,Y^n)-q_n(\bar{x}^n,Y^n)\right] $ exceeds $\delta$. 
 For Theorem \ref{th: rational theorem}, $\calA^c$ will signify the event of correct mismatched decoding, under which the same term can be lower bounded by the minimum value of $q_n(x^n,y^n)-q_n(\bar{x}^n,y^n)$ among all possible triplets of sequences $(x^n,\bar{x}^n,y^n)$. 
 
 We next prove Lemma \ref{lm: essential lemma}.
\begin{proof}
First note that choosing $a(x^n)=0$ for all $x^n$ rather than taking the supermum over $a(x^n)\geq 0$, we obtain the lower bound 
\begin{flalign}\label{eq: recall that 1}
&\sup_{ a(x^n)\geq 0}\mathbb{E}\left(\log\frac{e^{s_nq_n(X^n,Y^n)+a(X^n)}}{\sum_{\bar{x}^n}{P}^{(n)}(\bar{x}^n)e^{s_nq_n(\bar{x}^n,Y^n)+a(\bar{x}^n)}}\right)\nonumber\\
\geq &  \frac{1}{n}\mathbb{E}\left(\log\frac{e^{s_nq_n(X^n,Y^n)}}{\sum_{\bar{x}^n}{P}^{(n)}(\bar{x}^n)e^{s_nq_n(\bar{x}^n,Y^n)}}\right)\nonumber\\
=& \Pr\left\{\calA\right\}\cdot  \frac{1}{n}\mathbb{E}\left(\log\frac{e^{s_nq_n(X^n,Y^n)}}{\sum_{\bar{x}^n}{P}^{(n)}(\bar{x}^n)e^{s_nq_n(\bar{x}^n,Y^n)}}\big|\calA \right)+Pr\left\{\calA^c\right\}\cdot  \frac{1}{n}\mathbb{E}\left(\log\frac{e^{s_nq_n(X^n,Y^n)}}{\sum_{\bar{x}^n}{P}^{(n)}(\bar{x}^n)e^{s_nq_n(\bar{x}^n,Y^n)}}\big|\calA^c\right) 
,
\end{flalign}
where the last step follows from the law of total probability. Now, define
\begin{flalign}
&q_{\max}(Y^n)\triangleq \max_{x^n\in\calC_n}q_n(x^n,Y^n),\label{eq: qmx dfn}\\
&Z_n\triangleq  q_{\max}(Y^n)-q_n(X^n,Y^n)\label{eq: delta definition}
\end{flalign} 
and note that the boundedness assumption \eqref{eq: boundedness updated} implies that $Z_n\leq 2B$. Therefore,
\begin{flalign}
\log\frac{e^{s_nq_n(X^n,Y^n)}}{\sum_{\bar{x}^n}{P}^{(n)}(\bar{x}^n)e^{s_nq_n(\bar{x}^n,Y^n)}}\geq 
\log\frac{e^{s_nq_n(X^n,Y^n)}}{e^{s_nq_{\max}(Y^n)}}=s_n\left(q_n(X^n,Y^n)-q_{\max}(Y^n)\right)\geq -s_n\cdot 2B,
\end{flalign}
and consequently we can lower bound the first conditional expectation in (\ref{eq: recall that 1}) as follows
\begin{flalign}\label{eq: reviewers 2}
&  \mathbb{E}\left(\frac{1}{n}\left. \log\frac{e^{s_nq_n(X^n,Y^n)}}{\sum_{\bar{x}^n}{P}^{(n)}(\bar{x}^n)e^{s_nq_n(\bar{x}^n,Y^n)}}\right|\calA \right)\geq -\frac{s_n\cdot 2B}{n} .
\end{flalign}
As for the second conditional expectation in (\ref{eq: recall that 1}), 
\begin{flalign}\label{eq: reviewers 3}
&  \mathbb{E}\left(\frac{1}{n}\left. \log\frac{e^{s_nq_n(X^n,Y^n)}}{\sum_{\bar{x}^n}{P}^{(n)}(\bar{x}^n)
e^{s_nq_n(\bar{x}^n,Y^n)}} \right|\calA^c 
\right)\nonumber\\
=&    -\mathbb{E}\left(\frac{1}{n}\left. \log\left(P^{(n)}(X^n)+ \sum_{\bar{x}^n\in\calC_n:\; \bar{x}^n\neq X^n }{P}^{(n)}(\bar{x}^n) e^{-s_n\left[ q_n(\bar{x}^n,Y^n)- q_n(X^n,Y^n)\right]}\right)  \right|\calA^c \right)\nonumber\\
\stackrel{(*)}{=}&    -\mathbb{E}\left(\frac{1}{n}\left. \log\left(e^{-nR_n}+ \sum_{\bar{x}^n\in\calC_n:\; \bar{x}^n\neq X^n }{P}^{(n)}(\bar{x}^n) e^{-s_n\left[ q_n(\bar{x}^n,Y^n)- q_n(X^n,Y^n)\right]}\right)  \right|\calA^c \right)\nonumber\\
\geq&    -\mathbb{E}\left(\frac{1}{n}\left. \log\left(e^{-nR_n}+ e^{-s_n\min_{\bar{x}^n\in\calC_n:\; \bar{x}^n\neq X^n }\left[ q_n(\bar{x}^n,Y^n)- q_n(X^n,Y^n)\right]} \right)  \right|\calA^c \right)\nonumber\\
\geq &  \mathbb{E}\left( \min\left\{R_n, \frac{s_n}{n}\cdot \min_{\bar{x}^n\in\calC_n:\; \bar{x}^n\neq X^n }\left[ q_n(X^n,Y^n)-q_n(\bar{x}^n,Y^n)\right]  \right\}\bigg |\calA^c\right)-\frac{1}{n}\log(2) ,
\end{flalign}
where $(*)$ follows since $P^{(n)}$ is uniform over $\calC_n$, and the last step follows since $e^{-a}+e^{-b}\leq 2e^{-\min\{a,b\}}$.

\end{proof}

\subsection{An Equality Between $C_q^{(\infty)}$ and the Capacity with Margin Decoding}
The next theorem states that $C_q^{(\infty)}$ is equal to $C_{q,margin}^{const}(W)$, and implies that the following two statements are equivalent:
(a) $C_q(W)=C_q^{(\infty)}(W)$ and 
(b) $C_q(W)=C_{q,margin}^{const}(W)$.

\begin{theorem}\label{th: Cinfinity equals C margin}
For a metric $q$ which satisfies the boundedness assumption (\ref{eq: boundedness updated}) 
one has 
\begin{flalign}
C_q^{(\infty)}(W)=C_{q,margin}^{const}(W).
\end{flalign}
The theorem is established by proving the two inequalities, $C_q^{(\infty)}(W)\geq C_{q,margin}^{const}(W)$ and $C_q^{(\infty)}(W)\leq C_{q,margin}^{const}(W)$. 
The proof of the latter is based on Lemma \ref{lm: essential lemma} where the event $\calA$ in \eqref{eq: main lemma result} is chosen as the error event in margin decoding and with a proper choice of parameter $s_n$. 
 
\end{theorem}

\begin{proof}
\noindent \underline{Proof of $C_q^{(\infty)}(W)\leq C_{q,margin}^{const}(W)$}: 

\vspace{0.25cm}

The average probability of error associated with a random coding constant composition scheme within the type-class $T(P_X)$, and a $\delta$-margin decoder can be expressed as
\begin{flalign}\label{eq: tandom coding threshold}
\bar{P}_e=& 
\sum_{x^n\in T(P_X), y^n}\frac{1}{|T(P_X)|} W^n(y^n|x^n)\left[1 -\left(1-\sum_{\bar{x}^n\in T(P_X):\; q_n(\bar{x}^n,y^n)\geq q_n(x^n,y^n)-\delta} \frac{1}{|T(P_X)|} \right)^{M_n}\right] ,
\end{flalign}
where $1- \bar{P}_e$ stands for the probability of drawing a transmitted codeword whose metric with the channel output 
exceeds that of all the other independently drawn codewords by at least $\delta$.  
Optimizing $P_X$ and 
using standard method of types tools and the weak law of large numbers, one can deduce that for all $\epsilon>0$, the rate $C_{q,margin}^{(1)}(W,\delta)-\epsilon$ is achievable with a $\delta$-margin decoder, where 
\begin{flalign}\label{eq: Cq1 thresh dfn}
C_{q,margin}^{(1)}(W,\delta)=& \max_{P_X} \underset{P_{\widetilde{Y}|X}:\; P_{\widetilde{Y}}=P_{Y}, \mathbb{E}(q(X,\widetilde{Y}))\geq \mathbb{E}(q(X,Y))-\delta  }{\min}I(X;\widetilde{Y}),
\end{flalign}
with $P_X\in\calP(\calX)$, $(X,Y)\sim P_X\times W$, $(X,\widetilde{Y})\sim P_X\times P_{\widetilde{Y}|X}$, and where $P_Y,P_{\widetilde{Y}}$ are the corresponding marginal distributions of $Y,\widetilde{Y}$, respectively. Now, using the continuity of $C_{q,margin}^{(1)}(W,\delta)$ w.r.t.\ $\delta$ (which follows from its convexity, see e.g., \cite[Appendix VI]{CsiszarNarayan95}), we obtain $\lim_{\delta\downarrow 0} C_{q,margin}^{(1)}(W,\delta)= C_q^{(1)}(W)$, and as a result, the $\delta$-margin capacity is lower bounded by $C_q^{(1)}(W)$ and therefore also by $C_q^{(\infty)}(W)$. Consequently we obtain $C_{q,margin}^{const}(W)=\sup_{\delta>0}C_{q,margin}(W,\delta)\geq C_q^{(\infty)}(W)$.

\vspace{0.5cm}

\noindent\underline{Proof of $C_q^{(\infty)}(W)\geq C_{q,margin}^{const}(W)$}: 

\vspace{0.25cm}

Fix $\delta$ and a vanishing sequence $\epsilon_n,n\geq1$. 
Let $\calP_U(\epsilon_n,\delta)$ stand for the set of $P^{(n)}\in\calP(\calX^n)$ which are uniform over a subset $\calC_n$ of $\calX^n$ and such that $P_{margin}(W^{(n)},\calC_n,q_n,\delta)= \epsilon_n$; i.e., 
\begin{flalign}\label{eq: calPUTPX dfn}
&\calP_U(\epsilon_n,\delta)\nonumber\\
\triangleq& \left\{
P^{(n)}\in\calP(\calX^n):\; 
P^{(n)}(\tilde{x}^n) = P^{(n)}(x^n) \mbox{ if }P^{(n)}(x^n)\cdot P^{(n)}(\tilde{x}^n)>0, 
P_{margin}(W^{(n)},\calC_n,q_n,\delta)= \epsilon_n\right\},
\end{flalign}
where $\calC_n$ stands for the support of $P^{(n)}$.

Since $\calP_U(\epsilon_n,\delta)\subseteq\calP(\calX^n)$, 
\begin{flalign}
\label{eq: recall that 111}
C_q^{(n)}(W)\geq &  \frac{1}{n}\max_{P^{(n)}\in\calP_U(\epsilon_n,\delta)}\sup_{s_n\geq 0} \mathbb{E}\left(\log\frac{e^{s_nq_n(X^n,Y^n)}}{\sum_{\bar{x}^n}{P}^{(n)}(\bar{x}^n)e^{s_nq_n(\bar{x}^n,Y^n)}}\right)
.
\end{flalign}

Let $P^{(n)}\in \calP_U(\epsilon_n,\delta)$ be given, and denote its support by $\calC_n$. 
Next, let $\calE$ denote the error event 
\begin{flalign}
\calE=\left\{(x^n,y^n):\; 
q_n(x^n,y^n)<  \max_{\tilde{x}^n\in\calC_n :\; \tilde{x}^n\neq x^n}  q_n(\tilde{x}^n,y^n)+\delta\right\},
\end{flalign}
and note that $\Pr(\calE)= P_{margin}(W^{(n)},\calC_n,q_n,\delta)= \epsilon_n$. 

Now recall (\ref{eq: DUAL EXPRESSION MODIFIED}), we obtain for $X^n\sim P^{(n)}$ that is uniform over a subset $\calC_n\in \calX^n$ of size $|\calC_n|=e^{R_n}$, 
\begin{flalign}
&C_q^{(n)}(W)\nonumber\\
=& \frac{1}{n}\sup_{s_n\geq 0, a(x^n)\geq 0} \mathbb{E}\left(\log\frac{e^{s_nq_n(X^n,Y^n)+a(X^n)}}{\sum_{\bar{x}^n}{P}^{(n)}(\bar{x}^n)e^{s_nq_n(\bar{x}^n,Y^n)+a(\bar{x}^n)}}\right)\nonumber\\
\geq &\sup_{s_n\geq 0}- \Pr\left\{\calE\right\}\cdot \frac{s_n\cdot 2B}{n}  + \Pr\left\{\calE^c\right\}\cdot  \mathbb{E}\left( \min\left\{R_n, \frac{s_n}{n}\cdot \min_{\bar{x}^n\in\calC_n:\; \bar{x}^n\neq X^n }\left[ q_n(X^n,Y^n)-q_n(\bar{x}^n,Y^n)\right]  \right\}\bigg |\calE^c\right)\nonumber\\
&\quad-\frac{1}{n}\log(2)\nonumber\\
\geq& \sup_{s_n\geq 0} -\frac{\epsilon_ns_n\cdot 2B}{n}  +(1-\epsilon_n)\cdot \min\left\{R_n,\frac{s_n\delta}{n}\right\} -\frac{1}{n} \log2 
,
\end{flalign}
where the first inequality follows from (\ref{eq: main lemma result}) by substituting $\calA\Rightarrow \calE$ and the second in equality follows since under $\calE^c$, we have $\min_{\bar{x}^n\in\calC_n:\; \bar{x}^n\neq X^n }\left[ q_n(X^n,Y^n)-q_n(\bar{x}^n,Y^n)\right] \geq \delta$.

Taking\footnote{We choose $s_n=\frac{nR_n}{\delta}$ 
to balance the effect of the two terms $-\frac{\epsilon_n s_n\cdot 2B}{n}$ and $\min \{R_n,\frac{s_n\delta}{n}\}$, and to obtain an asymptotic lower bound of $R_n$. Any choice smaller than $nR_n/\delta$ would result in the second term being strictly smaller than $R_n$.} $s_n=\frac{nR_n}{\delta}$, 
it follows that
\begin{flalign}
\label{eq: recall that 111b}
C_q^{(n)}(W)\geq &\max_{P^{(n)}\in\calP_U(\epsilon_n,\delta)}
 -\frac{\epsilon_n\cdot 2B R_n}{\delta} +(1-\epsilon_n)\cdot R_n -\frac{1}{n} \log2.
\end{flalign}

Now, since this inequality holds for every $\delta>0$ and every non-negative vanishing sequence $\epsilon_n,n\geq 1$, we obtain
\begin{flalign}\label{eq: threshthresh 2}
C_q^{(\infty)}(W)\geq  
&\sup_{ \substack{\delta>0, \{\epsilon_n\}_{n\geq 1}:\; \epsilon_n\geq 0 ,\\  \lim_{n\rightarrow\infty} \epsilon_n=0 }}\limsup_{n\rightarrow\infty} \max_{P^{(n)}\in \calP_U(\epsilon_n,\delta)}   -\frac{\epsilon_n\cdot 2B R_n}{\delta} + (1-\epsilon_n)R_n\nonumber\\
=& C_{q,margin}^{const}(W),\end{flalign}
where last step follows since the supremum of $R$ over $\delta$ such that there exists a sequence of codebooks $\calC_n,n\geq 1$ with $\frac{1}{n}\log |\calC_n|\geq R$ and which incur $\lim_{n\rightarrow\infty}\epsilon_n=0$ is nothing but $C_{q,margin}^{const}(W)$. 
This concludes the proof of the inequality $C_q^{(\infty)}(W)\geq C_{q,margin}^{const}(W)$ and Theorem \ref{th: Cinfinity equals C margin} follows. 
\end{proof}

A special case of a decoding metric is the erasures-only metric
\begin{flalign}
q_{e.o.}(x,y)=1\{W(y|x)>0\}.
\end{flalign}
This metric has the property that no undetected errors can occur, since by definition the transmitted codeword always accumulates a $1$-value metric, which is the highest possible. 
Since in the erasures-only metric case, it has been shown (\hspace{1sp}{\cite{CsiszarNarayan95}}) that the mismatch capacity is equal to the threshold capacity and to $C_q^{(\infty)}(W)$, we deduce from Theorem \ref{th: Cinfinity equals C margin} that in this case the mismatch capacity with a constant margin decoder is equal to the mismatch capacity.

Note that the proof of Theorem \ref{th: Cinfinity equals C margin} implies that taking $a(x^n)=0$ in (\ref{eq: DUAL EXPRESSION MODIFIED}) and maximizing over $P^{(n)}$ incurs a loss of rate (compared to the optimal value of non-constant $a(x^n)$) which vanishes as $n$ tends to infinity. 

 To see this, let  \begin{flalign}\label{eq: DUAL EXPRESSION MODIFIED with a0}
\widetilde{C}_q^{(n)}(W)=& \frac{1}{n}\max_{P^{(n)}\in\calP(\calX^n)}\sup_{s_n\geq 0} \mathbb{E}\left(\log\frac{e^{s_nq_n(X^n,Y^n)}}{\sum_{\bar{x}^n}{P}^{(n)}(\bar{x}^n)e^{s_nq_n(\bar{x}^n,Y^n))}}\right),
\end{flalign}
and observe that the first step (\ref{eq: recall that 1}) of the proof is to lower bound $C_q^{(n)}(W)$ by $\widetilde{C}_q^{(n)}(W)$, so the same proof actually gives: 
\begin{flalign}
C_q^{(\infty)}(W)\geq \widetilde{C}_q^{(\infty)}(W)\geq C_{q,margin}^{const}(W).
\end{flalign}
On the other hand, it is also proven that $C_q^{(\infty)}(W)\leq C_{q,margin}^{const}(W)$ so we have:
\begin{flalign}\label{eq: Sandwich}
C_q^{(\infty)}(W)\geq \widetilde{C}_q^{(\infty)}(W)\geq C_{q,margin}^{const}(W)\geq C_q^{(\infty)}(W).
\end{flalign}
This sandwich argument shows that $C_q^{(\infty)}(W)= \widetilde{C}_q^{(\infty)}(W)$. 
 Indeed, for finite $n$ there is usually a strict inequality $\widetilde{C}_q^{(n)}(W)< C_q^{(n)}(W)$, as it is known that the generalized mutual information (GMI) \cite{MerhavKaplanLapidothShamai94} can be strictly smaller than the LM rate.
 Nevertheless, in the limit as $n$ tends to infinity, they are equal. The explanation we provide is that without loss of asymptotic optimality the maximizing $P^{(n)}$ (in the dual expression (\ref{eq: DUAL EXPRESSION MODIFIED})) can be replaced by a uniform distribution over a subset of a single type-class in $\calX^n$, denoted $\calC$. Asymptotically, this may cancel the need to assign different $a(\cdot)$ values (in the dual expression (\ref{eq: DUAL EXPRESSION MODIFIED})) to different sequences $x^n$, at least in the bounded metric case.

\subsection{A Relationship Between Constant Margin decoding and Constant Threshold Decoding}\label{sc: Threshold-Margin relationship} 

In this section we establish a connection between the capacity with constant threshold decoding and decoding with a constant margin decoder. 

\begin{definition}
Let $P\in\calP(\calX)$ be given. 
A sequence of constant-composition codebooks $\calC_n, n=1,2,...$ is said to be {\it $P$-constant composition}, if for all $n$, the codewords lie in a single type-class which corresponds to an empirical distribution $\hat{P}_n$ that is $o(n^{-1/2})$ close in the $L_1$ distance sense to $P$; i.e.,
\begin{flalign}\label{eq: convergence rate of Phat}
\limsup_{n\rightarrow \infty }\frac{ \|\hat{P}_n-P\|_1}{n^{-1/2} }=\limsup_{n\rightarrow \infty }\frac{\sum_{ x\in\calX} |\hat{P}_n(x)-P(x)|}{n^{-1/2} } =0. \end{flalign}
\end{definition}

We next define the constant threshold capacity and the constant margin capacity which are constrained to sequences of $P$-constant composition codes. 
\begin{definition}
The $(P,q)$-threshold capacity, denoted $C_{q,thresh}^{const}(P,W)$ is the supremum of achievable rates using sequences of $P$-constant composition codes, metic $q$, and a constant threshold decoder. \end{definition}

\begin{definition}
The $(P,q)$-margin capacity, denoted $C_{q,margin}^{const}(P,W)$ is the supremum of achievable rates using sequences of $P$-constant composition codes, metic $q$, and a constant margin decoder. \end{definition}

\begin{theorem}\label{th: thresh-margin theorem}
If $P\in\calP(\calX)$ is such that $\forall x\in\calX, P(x)>0$, $q$ is bounded (\ref{eq: bounded q  aaa}), and there exists at least one symbol $x^*$ which satisfies $\max_{y,y'\in\calY }W(y|x^*)q(x^*,y)-W(y'|x^* ) q(x^*,y')\neq 0$, 
then\begin{flalign}
C_{q,threshold}^{const}(P,W)\leq C_{q,margin}^{const}(P,W)\leq C_q^{(\infty)}(W).\label{eq: thresh margin c infty}
\end{flalign}
\end{theorem}
Note that the more general opposite inequality $C_{q,threshold}^{const}(W)\geq C_q^{(\infty)}(W)$ has been known \cite{CsiszarNarayan95} and is proved similarly to the proof of the claim 
$C_q^{(\infty)}(W)\leq C_{q,margin}^{const}(W)$ (replacing $q_n(x^n,y^n)$ in (\ref{eq: tandom coding threshold}) by $\mathbb{E}(q(X,Y))-\tau$). 

It would have been desirable to answer the more general question whether $C_{q,thresh}^{const}(W)\leq C_q^{(\infty)}(W)$ or not. 
Clearly, this result would follow if $C_{q,thresh}^{const}(W)=\sup_{P\in\calP(\calX):\;\forall x\in\calX, P(x)>0} C_{q,threshold}^{const}(P,W)$. 
It is important to note that the requirement that the sequence of codebooks be constant-composition does not limit the generality\footnote{The generality is not limited since the number of type-classes grows polynomially with $n$, whereas the number of codewords grows exponentially with $n$.}. The requirement that might limit the generalization of the result is (\ref{eq: convergence rate of Phat}); that the empirical distribution converges to some $P$ at a rate which is $o(n^{-1/2})$.   
 The mismatch capacity as well as the $\delta$-margin capacity are insensitive to small fluctuations in the type class in which the codewords lie. 
 This is because the decision regions and the corresponding error probabilities dictated by (\ref{eq: decoder decision rule}) and 
(\ref{eq: margin decision rule 1}) do not change if each codeword of a codebook is concatenated by a constant identical string of symbols 
 of length $f_n$ where $\lim_{n\rightarrow\infty}f_n/n=0$, this simply adds the same quantity to both sides of the equations (\ref{eq: decoder decision rule}) and 
(\ref{eq: margin decision rule 1}). The threshold capacity with a given threshold sequence $\tau_n$ is more sensitive to such fluctuations, but we conjecture that this issue is minor and does not have a significant effect on the generality of the results.  

The requirements on $P$ in Theorem \ref{th: thresh-margin theorem} are needed to invoke the Central Limit Theorem (CLT) which gives a strictly positive non-vanishing lower bound on $\Pr\left\{ q_n(X^n,Y^n)<\mathbb{E}_{\hat{P}_n\times W}(q(X,Y))-kn^{-1/2}\right\}$. This forces the threshold decoder to use a threshold level that is lower than $\mathbb{E}_{\hat{P}_n\times W}(q(X,Y))-kn^{-1/2}$. 

As mentioned above, if $q=q_{e.o}$ the theorem holds, hence we can assume $q\neq q_{e.o}$. The condition that there exists at least one symbol $x^*$ which satisfies $\max_{y,y'\in\calY }W(y|x^*)q(x^*,y)-W(y'|x^* ) q(x^*,y')\neq 0$ is not restrictive because it is not hard to realize that the result $C_q^{(\infty)}(W)=C_q(W)$ for $q=q_{e.o}$ (\hspace{1sp}\cite{CsiszarNarayan95}) extends to the case of a metric which satisfies $\max_{y,y'\in\calY }|W(y|x)q(x,y)-W(y'|x) q(x,y')|=0,\forall x$ as well for constant composition codes. 

We now prove Theorem \ref{th: thresh-margin theorem}. 

\begin{proof}
The right-most inequality in (\ref{eq: thresh margin c infty}) is a direct consequence of Theorem \ref{th: Cinfinity equals C margin}. 
The rest of the proof relies on the following two lemmas. 
\begin{lemma}\label{cl: Claim 1}
 Let $P\in\calP(\calX)$  which satisfies the conditions of Theorem \ref{th: thresh-margin theorem} and a sequence $\calC_n, n=1,2,...$ of $P$-constant composition codes be given where $\calC_n\subseteq T(\hat{P}_n)$. Let $X^n$ denote the random codeword which is uniformly distributed over $\calC_n$, and let $Y^n$ be the corresponding output of the channel $W$, then  
\begin{flalign}\label{eq: claim 1 equation}
\liminf_{n\rightarrow\infty}\Pr\left\{ q_n(X^n,Y^n)<\mathbb{E}_{\hat{P}_n\times W}(q(X,Y))-kn^{-1/2}\right\}> 0, \forall 0<k<\infty. 
\end{flalign}
\end{lemma}
The proof of Lemma \ref{cl: Claim 1} is based on the CLT and appears in Appendix \ref{ap: Proof of Claim 1}.

\begin{lemma}\label{cl: Claim 2}
Let $X^n,Y^n$ be defined as in Lemma \ref{cl: Claim 1}, then for all $\epsilon>0$,  
\begin{flalign}\label{eq: claim 2 result}
&P_{margin}(W^{(n)},\calC_n,q_n,\epsilon)\nonumber\\
\leq& 
P_{thresh}(W^{(n)},\calC_n,(q_n,\tau))+ \Pr\left\{ q_n(X^n,Y^n) \leq \tau+\epsilon\right\}. 
\end{flalign}
\end{lemma}
The proof of Lemma \ref{cl: Claim 2} appears in Appendix \ref{ap: Proof of Claim 2}.

Equipped with the lemmas we can now prove the theorem. 
From the condition (\ref{eq: convergence rate of Phat}) and since $|\mathbb{E}_{\hat{P}_n\times W}(q(X,Y))- \mathbb{E}_{P\times W}(q(X,Y))|\leq  (q_{max}-q_{min})\cdot \|\hat{P}_n-P\|_1$, where $q_{min}=\min_{x,y:\; W(y|x)>0}q(x,y)$, $q_{max}=\max_{x,y:\; W(y|x)>0}q(x,y)$, we have for all $n$ sufficiently large, 
\begin{flalign}\label{eq: q what diff}
\left| \mathbb{E}_{P\times W}(q(X,Y))- \mathbb{E}_{\hat{P}_n\times W}(q(X,Y)) \right| \leq (q_{max}-q_{min})n^{-1/2}.
\end{flalign}
Therefore,  for all $n$ sufficiently large, 
\begin{flalign}
&\Pr\left\{ q_n(X^n,Y^n)<\mathbb{E}_{P\times W}(q(X,Y))\right\}\nonumber\\
\geq &\Pr\left\{ q_n(X^n,Y^n)<\mathbb{E}_{\hat{P}_n\times W}(q(X,Y)) -(q_{max}-q_{min})n^{-1/2}\right\} .
\end{flalign}

We conclude from (\ref{eq: claim 1 equation}) that 
\begin{flalign}
&\liminf_{n\rightarrow\infty}\Pr\left\{ q_n(X^n,Y^n)<\mathbb{E}_{P\times W}(q(X,Y))\right\}\nonumber\\
\geq & \liminf_{n\rightarrow\infty}\Pr\left\{ q_n(X^n,Y^n)<\mathbb{E}_{\hat{P}_n\times W}(q(X,Y)) -(q_{max}-q_{min})n^{-1/2}\right\}\nonumber\\
 >&0.
\end{flalign}
Since $\liminf_{n\rightarrow\infty}\Pr\left\{ q_n(X^n,Y^n)<\tau\right\}=0$ is a necessary condition for a vanishing probability of error in $\tau$ threshold decoding, 
we deduce that for a sequence of $P$-constant composition codebooks, the {\it constant} threshold level $\tau$ should be strictly lower than $\mathbb{E}_{P\times W}(q(X,Y))$. 
Hence, denote $\tau=\mathbb{E}_{P\times W}(q(X,Y))-2\epsilon$ where $\epsilon>0$,  
thus from (\ref{eq: claim 2 result}) we obtain
\begin{flalign}\label{eq: rejghsli uversgf}
&P_{margin}(W^{(n)},\calC_n,q_n,\epsilon)\nonumber\\
\leq& 
P_{thresh}(W^{(n)},\calC_n,(q_n,\tau))+ \Pr\left\{ q_n(X^n,Y^n) \leq \mathbb{E}_{P\times W}(q(X,Y))-\epsilon\right\}. 
\end{flalign}
We will next show that the right most term $\Pr\left\{ q_n(X^n,Y^n) \leq \mathbb{E}_{P\times W}(q(X,Y))-\epsilon\right\}$ vanishes as $n$ tends to infinity, which will imply that if $P_{thresh}(W^{(n)},\calC_n,(q_n,\tau))$ vanishes then so does $P_{margin}(W^{(n)},\calC_n,q_n,\epsilon)$. 

Define
\begin{flalign}\label{eq: Ax dfn}
n_x=&\sum_{i=1}^n 1\{X_i=x\} \nonumber\\
A_x= &\frac{1}{n_x}\sum_{i=1}^n 1\{X_i=x\} \cdot q(x, Y_i)
\end{flalign}
and note that 
\begin{flalign}\label{eq: Ax dfn 2}
q_n(X^n,Y^n)=\sum_{x=1}^{|\calX|}\frac{n_x}{n}A_x,
\end{flalign} 
and that $n_x$ is a deterministic quantity since all $x^n$'s lie in the same type-class.

By assumption, $P(x)>0,\forall x\in\calX$, from (\ref{eq: convergence rate of Phat}) we have that there exists $\zeta>0$ such that $A_x$ is a average of at least $n\zeta$ non deterministic i.i.d. random variables. 
This implies that $q_n(X^n,Y^n)$ is the weighted average of $|\calX|$ independent random variables $A_x,x\in\calX$, each of which is the average of at least $n\cdot \zeta $ non deterministic i.i.d, random variables. 
Thus, from (\ref{eq: convergence rate of Phat}) we obtain for sufficiently large $n$, 
\begin{flalign}\label{eq: this label}
&\Pr\left\{ q_n(X^n,Y^n) \leq \mathbb{E}_{P\times W}(q(X,Y))-\epsilon\right\}\nonumber\\
\stackrel{(a)}{\leq}  &\Pr\left\{ q_n(X^n,Y^n) \leq \mathbb{E}_{\hat{P}_n\times W}(q(X,Y))-\epsilon/2\right\}\nonumber\\
\stackrel{(b)}{\leq} & \Pr\left\{ \cup_{x=1}^{|\calX|} \left\{A_x \leq \mathbb{E}_{ W(\cdot|x)}(q(x,Y))-\epsilon/2\right\}\right\}\nonumber\\
\leq &\sum_{x=1}^{|\calX|} \Pr\left\{ A_x \leq \mathbb{E}_{ W(\cdot|x)}(q(x,Y))-\epsilon/2\right\},
\end{flalign}
where $(a)$ follows from (\ref{eq: q what diff}), $(b)$ follows from (\ref{eq: Ax dfn 2}) and since for a p.m.f.\ $\mu\in\calP(\calX)$, if $ \sum_{x=1}^{|\calX|} \mu(x)a_x\leq \sum_{x=1}^{|\calX|}\mu(x) b_x$, then there must exist at least one $x$ such that $a_x\leq b_x$, and the last step follows from the union bound. 
From the weak law of large numbers we have that the r.h.s.\ of (\ref{eq: this label})  vanishes as $n$ tends to infinity.
 Thus, from (\ref{eq: rejghsli uversgf}) it follows that if $P_{thresh}(W^{(n)},\calC_n,(q_n,\tau))$ vanishes as $n$ tends to infinity, then so does $P_{margin}(W^{(n)},\calC_n,q_n,\epsilon)$. 
Consequently, under the assumptions of Theorem \ref{th: thresh-margin theorem}, the supremum of the achievable rates with a constant margin decoder is at least as high as that of a constant threshold decoder and Theorem \ref{th: thresh-margin theorem} follows. 
\end{proof}
Note that in fact the proof of Theorem \ref{th: Cinfinity equals C margin} implies a stronger result, that is stated in the following corollary. 
\begin{corollary}\label{cr: thresh-margin theorem stronger}
Let $\delta_n$ be a vanishing sequence which satisfies $\lim_{n\rightarrow\infty}\delta_n n^{1/2}=\infty$. 
 If $P$ satisfies the conditions of Theorem \ref{th: thresh-margin theorem}, 
then for all $P$-constant composition sequences of codebooks, 
the supremum of the achievables rates with a $\delta_n$-margin decoder is at least as high as that of any time varying $\tau_n$ threshold decoder. 
\end{corollary}
\begin{proof}
The proof goes along the lines of the proof of Theorem \ref{th: thresh-margin theorem} where in (\ref{eq: rejghsli uversgf}) one can substitute $\epsilon$ with a sequence $\delta_n$ which vanishes more slowly than $O(n^{-1/2})$. 
This is because in fact \begin{flalign}\label{eq: tau n must satisfy 2}
\liminf_{n\rightarrow\infty} \frac{\tau_n -(\mathbb{E}_{P\times W}(q(X,Y))-n^{-1/2})}{n^{-1/2}}<0  
\end{flalign}
should hold, so without loss of generality one can set $\tau_n=\mathbb{E}_{P\times W}(q(X,Y)-2\delta_n $ and 
then from the CLT, we shall still have that the r.h.s.\ of (\ref{eq: this label}) vanishes as $n$ tends to infinity. 

\end{proof}

\subsection{$C_{q}^{(\infty)}(W)$ is the Highest Achievable Rate for a Sufficiently Fast Decay of the Average Probability of Error}

Let $\eta_n$ denote the smallest non-zero metric difference among all possible codewords and channel outputs; that is, 
\begin{flalign}\label{eq: eta dfn}
\eta_n=&\min_{x^n,\tilde{x}^n,y^n:\; q_n(\tilde{x}^n,y^n)\neq q_n(x^n,y^n)} \left| q_n(\tilde{x}^n,y^n)- q_n(x^n,y^n)\right|,
\end{flalign}
where the minimization is over triplets of $n$-vectors in $\calX^n\times\calX^n\times \calY^n$. 

Note that if the set $\left\{x^n,\tilde{x}^n,y^n:\; q_n(\tilde{x}^n,y^n)\neq q_n(x^n,y^n) \right\}$ is non-empty, $\eta_n>0$ (with strict inequality), and if it is empty for any $n$, this means that the mismatch capacity of the channel is zero, since any two sequences $x^n,\tilde{x}^n$ are indistinguishable using the metric $q$ for every channel output. 

The next theorem states that 
$C_{q}^{(\infty)}(W)$ is the highest achievable rate for sequences of codes whose average probability of error satisfies $\lim_{n\rightarrow\infty}\frac{1}{\eta_n}P_e(W^n,\calC_n,q_n)=0$. 
This
implies a soft converse as an important special case when the metric $q$ is rational; $C_{q}^{(\infty)}(W)$ is the highest achievable rate for code sequences whose average probability of error converges to zero faster than $O(1/n)$.  

\begin{theorem}\label{th: rational theorem}
For a metric $q$ which satisfies the boundedness assumption (\ref{eq: boundedness updated}), every code-sequence $\left\{\calC_n\right\}_{n\in \mathbb{N}}$ for which $\lim_{n\rightarrow\infty}\frac{1}{\eta_n}P_e(W^n,\calC_n,q_n)=0$ must satisfy 
\begin{flalign}
R=\mbox{limsup}_{n\rightarrow\infty} \frac{1}{n}\log|\calC_n| \leq C_q^{(\infty)}(W).
\end{flalign} 
In particular, if 
$q$ is a rational metric and $\lim_{n\rightarrow\infty} n\cdot P_e(W^n,\calC_n,q_n)=0$, then $R\leq C_q^{(\infty)}(W)$. 
\end{theorem}
Theorem \ref{th: rational theorem} constitutes a step forward towards verifying whether Csisz\'{a}r and Narayan's conjecture \cite{CsiszarNarayan95} holds; i.e., whether $C_q(W)=C_q^{(\infty)}(W)$. For rational metrics, it identifies $C_q^{(\infty)}(W)$ as the highest rate achievable with average probability of error which is $o(1/n)$. The gap can be closed if one could show that the average probability of error vanishes exponentially fast (or at least as fast as $o(1/n)$) at any rate below capacity. 
Since at rates below $C_q^{(\infty)}(W)$ exponential decay of the average probability of error is feasible \cite{CsiszarNarayan95}, the following corollary can be stated.
\begin{corollary}\label{th: rational corollary}
For a bounded rational metric $q$, $C_q(W)= C_q^{(\infty)}(W)$ iff for all $R<C_q(W)$ there exists a sequence of codebooks $\{\calC_n\}_{n\in \mathbb{N}}$, such that $\liminf_{n\rightarrow\infty}\frac{1}{n}\log|\calC_n|\geq R$ and $\liminf_{n\rightarrow\infty}-\frac{1}{n}\log P_e(W^n,\calC_n,q_n)>0$. 
\end{corollary}
We now prove Theorem \ref{th: rational theorem}. 
\begin{proof}
Fix a vanishing sequence $\tilde{\epsilon}_n,n\geq1$. 
Let $\calP_U(\tilde{\epsilon}_n)$ stand for the set of $P^{(n)}\in\calP(\calX^n)$ which are uniform over a subset $\calC_n$ of $\calX^n$ 
and such that $P(W^{(n)},\calC_n,q_n)= \tilde{\epsilon}_n$; i.e., 
\begin{flalign}\label{eq: calPUTPX dfn copied}
&\calP_U(\tilde{\epsilon}_n)\nonumber\\
\triangleq& \left\{
P^{(n)}\in\calP(\calX^n):\; 
P^{(n)}(\tilde{x}^n) = P^{(n)}(x^n) \mbox{ if }P^{(n)}(x^n)\cdot P^{(n)}(\tilde{x}^n)>0, 
P(W^{(n)},\calC_n,q_n)= \tilde{\epsilon}_n\right\},
\end{flalign}
where $\calC_n$ stands for the support of $P^{(n)}$. 

Let $P^{(n)}\in \calP_U(\tilde{\epsilon}_n)$ be given, and denote by $\calC_n$ 
the support set of $P^{(n)}$, which will be referred to as a codebook. Denote again $R_n=\frac{1}{n}\log |\calC_n|$. 

Let $\widetilde{\calE}$ be the error event; that is, 
\begin{flalign}
\widetilde{\calE}=& \left\{(x^n,y^n):\; q_n(x^n,y^n)\leq \max_{\tilde{x}^n\in\calC_n:\tilde{x}^n\neq x^n}q_n(\tilde{x}^n,y^n) \right\}, 
\end{flalign}
and note that $\Pr(\widetilde{\calE})= P_e(W^{(n)},\calC_n,q_n)=\tilde{ \epsilon}_n$. 

Next we invoke Lemma \ref{lm: essential lemma} with $\calA=\widetilde{\calE}$. 
Recall (\ref{eq: DUAL EXPRESSION MODIFIED}), we obtain
\begin{flalign}
& \frac{1}{n}\sup_{s_n\geq 0, a(x^n)\geq 0} \mathbb{E}\left(\log\frac{e^{s_nq_n(X^n,Y^n)+a(X^n)}}{\sum_{\bar{x}^n}{P}^{(n)}(\bar{x}^n)e^{s_nq_n(\bar{x}^n,Y^n)+a(\bar{x}^n)}}\right)\nonumber\\
\geq &\sup_{s_n\geq 0}- \Pr\left\{\widetilde{\calE}\right\}\cdot \frac{s_n\cdot 2B}{n}  + \Pr\left\{\widetilde{\calE}^c\right\}\cdot  \mathbb{E}\left( \min\left\{R_n, \frac{s_n}{n}\cdot \min_{\bar{x}^n\in\calC_n:\; \bar{x}^n\neq X^n }\left[ q_n(X^n,Y^n)-q_n(\bar{x}^n,Y^n)\right]  \right\}\bigg |\widetilde{\calE}^c\right)\nonumber\\
&\quad-\frac{1}{n}\log(2)\nonumber\\
\geq& \sup_{s_n\geq 0} -\frac{\tilde{\epsilon}_ns_n\cdot 2B}{n}  +(1-\tilde{\epsilon}_n)\cdot \min\left\{R_n,\frac{s_n\eta_n}{n}\right\} -\frac{1}{n} \log2 
,
\end{flalign}
where the first inequality follows from (\ref{eq: main lemma result}) by substituting $\calA\Rightarrow \widetilde{\calE}$ and the second in equality follows by definition of $\eta_n$ (\ref{eq: eta dfn}) and by the fact that under $\widetilde{\calE}^c$, $q_n(X^n,Y^n)-\max_{\bar{x}^n\in\calC_n:\; \bar{x}^n\neq X^n }q_n(\bar{x}^n,Y^n)\geq \eta_n$. 

Finally we substitute\footnote{We choose $s_n=\frac{nR_n}{\eta_n}$ 
for the term $\min \{R_n,\frac{s_n\delta}{n}\}$ to yield $R_n$ which is our desired lower bound.} $s_n=\frac{nR_n}{\eta_n}$, and take the supremum over the sequences $\tilde{\epsilon}_1,\tilde{\epsilon}_2,...$ which satisfy $\lim_{n\rightarrow\infty} \frac{\tilde{\epsilon}_n}{{\eta}_n}=0$, 
 this yields 
\begin{flalign}\label{eq: threshthresh 2 copied}
C_q^{(\infty)}(W)\geq  
&\sup_{ \substack{ \{\tilde{\epsilon}_n\}_{n\geq 1}:\; \tilde{\epsilon}_n\geq 0 ,\\  \lim_{n\rightarrow\infty} \frac{\tilde{\epsilon}_n}{\eta_n}=0 }}\limsup_{n\rightarrow\infty} 
\max_{P^{(n)}\in \calP_U(\tilde{\epsilon}_n)}   -\frac{\tilde{\epsilon}_nR_n\cdot 2B }{\eta_n} + R_n.
\end{flalign}
Now, the r.h.s.\ of (\ref{eq: threshthresh 2 copied}) is nothing but the supremum of the rates for which $\lim_{n\rightarrow\infty} \frac{\tilde{\epsilon}_n}{\eta_n}=0 $ 
and this concludes the proof of the first statement of Theorem \ref{th: rational theorem}.

To prove the second statement of the theorem pertaining to rational metrics it suffices to prove that for a rational metric $q$ with denominator $D$ one has $\eta_n\geq \frac{1}{nD}$ as the following lemma states.
\begin{lemma}
If $q$ is a rational metric with denominator $D$, then 
\begin{flalign}
&\max_{\tilde{x}^n, x^n,y^n:\; q_n(\tilde{x}^n,y^n)\neq q_n(x^n,y^n)}|q_n(\tilde{x}^n,y^n)- q_n(x^n,y^n)|\geq\frac{1}{nD}. 
\end{flalign}
\end{lemma}
\begin{proof}
Since the metric is rational with denominator $D$, we can write
\begin{flalign}
q_n(\tilde{x}^n,y^n)=&
\frac{1}{n}\sum_{i=1}^n \frac{a_i}{D}\quad  ;\quad q_n(x^n,y^n)=\frac{1}{n}\sum_{i=1}^n \frac{b_i}{D},
\end{flalign}
where $a_i,b_i, i=1,...,n$, are integers. Thus, if $q_n(\tilde{x}^n,y^n)\neq q_n(x^n,y^n)$, assume w.l.o.g.\ that $q_n(\tilde{x}^n,y^n)> q_n(x^n,y^n)$, and we obtain
\begin{flalign}
\sum_{i=1}^n a_i >\sum_{i=1}^n b_i  .
\end{flalign}
Consequently, we deduce that 
\begin{flalign}
\sum_{i=1}^n a_i \geq 1+\sum_{i=1}^n b_i  ,
\end{flalign}
which yields
\begin{flalign}
\sum_{i=1}^n \frac{a_i}{nD} \geq \frac{1}{nD}+\sum_{i=1}^n \frac{b_i}{nD}  ,
\end{flalign}
and the desired result follows. 
\end{proof}
This concludes the proof of Theorem \ref{th: rational theorem}.
\end{proof}

\section{A Max-Min Upper Bound}\label{sc: Max-Min Upper Bounds}
We next present an upper bound on $C_q(W)$
which differs from the lower bound $C_q^{(\infty)}(W)$ (\ref{eq: C q infty dfn}) by the sets over which the inner minimizations are performed. 
\begin{theorem}\label{th: second theorem}
The mismatch capacity of the DMC $W$ with an additive metric $q$ satisfies
\begin{flalign}\label{eq: bound of Theorem 1 a}
&C_q(W)\leq  \liminf_{n\rightarrow\infty} \max_{P_{X^n}} \min_{\substack{P_{\widetilde{Y}^n|X^n }:\; \\ P_{\widetilde{Y}^n,q_n(X^n,\widetilde{Y}^n) }=P_{Y^n,q_n(X^n,Y^n) } } }
\frac{1}{n}  I(X^n;\widetilde{Y}^n),
\end{flalign}
where $P_{\widetilde{Y}^n,q_n(X^n,\widetilde{Y}^n) }$ is the joint distribution of $\widetilde{Y}^n$ and the (random) metric value $q_n(X^n,\widetilde{Y}^n)$ at the output of the channel $P_{\widetilde{Y}^n|X^n }$, and $P_{Y^n,q_n(X^n,Y^n) }$ is defined as for the channel $W^n$. 
\end{theorem}
\begin{proof}
The proof is similar to the proof of \cite[Theorem 5]{SomekhBaruch_general_formula_IT2015}. 

Let \begin{flalign}
\Phi_{q_n}\triangleq \mbox{Pr}\left\{q_n(\tilde{X}^n,Y^n)\geq q_n(X^n,Y^n)|X^n,Y^n\right\}
\end{flalign} 
where $X^n$ and $Y^n$ are the input and output channel vectors, respectively, and $\tilde{X}^n$ is independent of $(X^n,Y^n)$ and is distributed identically to $X^n$.

It was noted in \cite{SomekhBaruch_general_formula_IT2015} that for every codebook $\calC_n$, one has $P_e(W^n, \calC_n,q_n)=P_e(\widetilde{W}^{(n)}, \calC_n,q_n)$ for every channel $\widetilde{W}^{(n)}$ whose output $\widetilde{Y}^n$ shares the same joint law with $q_n(X^n,\widetilde{Y}^n)$ as that of $(Y^n,q_n(X^n,Y^n))$.
This is evident from the expression 
\begin{flalign}\label{eq: VerduHan UB here too}
&P_e(W^{(n)},\calC_n,q_n)\nonumber\\
=&\Pr\left\{-\frac{1}{n}\log \Phi_{q_n}< \frac{1}{n}\log M_n \right\},
\end{flalign}
which was derived in \cite[Lemma 1]{SomekhBaruch_general_formula_IT2015}.

The r.h.s.\ of (\ref{eq: VerduHan UB here too}) is identical for both channels $W^{(n)}$ and $\widetilde{W}^{(n)}$. 
Consequently, if $P_{\widetilde{Y}^n,q_n(X^n,\widetilde{Y}^n)}=P_{Y^n,q_n(X^n,Y^n)}$ for all $n$, the capacity of the channel $W^n$ is equal to that of the channel $\left\{\widetilde{W}^{(n)}\right\}_{n\in \mathbb{N}}$, whose mismatch capacity can in turn, be bounded essentially by the normalized mutual information $\frac{1}{n}I(X^n;\widetilde{Y}^n)$ yielding (\ref{eq: bound of Theorem 1 a}).

\end{proof}

\section{Conclusion}\label{sc: Conclusion}
In this work we considered a DMC with decoders that employ a mismatched metric. 
Finding a single-letter expression for the mismatch capacity of the DMC remains an open problem. Nevertheless, inspired by Csisz\'{a}r and Narayan's conjecture which states that the ``product-space" improvement of the random coding lower bound on the mismatch capacity $C_q^{(\infty)}(W)$ is the channel capacity, we studied this quantity. 
The significance of Csisz\'{a}r and Narayan's conjecture is threefold because if true, it implies that: (a) one can approach capacity by employing random coding schemes with product space input alphabets of increasing dimension, (b) the mismatch capacity is equal to the threshold capacity. (c) exponential decay of the average probability of error at rate $R-\epsilon$ for all $\epsilon>0$ is feasible.  

Several properties of $C_q^{(\infty)}(W)$ were proved in this paper. The first property is that $C_q^{(\infty)}(W)$ is equal to the mismatched capacity with a constant margin level. 
It is known \cite{CsiszarKorner81graph},\cite{Hui83} that $C_q^{(\infty)}(W)$ is upper bounded by the constant threshold capacity. In this paper we derived the opposite inequality for the case of $P$-constant compositions codes. We proved that the highest achievable rate with a constant margin decoder is equal to the highest achievable rate with a constant threshold decoder for sequences of $P$-constant compositions codes. Consequently, if such sequences of codes are used with constant threshold decoding, the highest achievable rate is upper bounded by $C_q^{(\infty)}(W)$. 
The second property states that if the average probability of error of a sequence of codebooks, $\epsilon_n$, converges  to zero sufficiently fast, in particular if $\epsilon_n/\eta_n$ vanishes (where $\eta_n$ is the minimal non-zero difference among all possible metric values between sequences of length $n$), 
 the rate of the code-sequence is upper bounded by 
$C_q^{(\infty)}(W)$. Consequently, in particular, we have shown that if $q$ is a bounded rational metric $\eta_n=O(1/n)$, and thus in this case, if the average probability of error converges to zero faster than $O(1/n)$, then $R\leq C_q^{(\infty)}(W)$. 
It therefore can be inferred that in this case if a sequence of codes of rate $R$ is known to achieve an average probability of error which is $o(1/n)$, there exists a sequence of codes operating at a rate arbitrarily close to $R$ with an average probability of error that vanishes exponentially fast. 
It would be an interesting problem to specify a lower bound on the convergence rate of  $\eta_n$ to zero as $n$ tends to infinity for non rational metrics, which will lead to more general result. 

Although we could not verify whether Csisz\'{a}r and Narayan's conjecture holds, the results obtained in this paper increase our understanding of the quantity $C_q^{(\infty)}(W)$ and constitute a step forward towards verifying whether the conjecture holds.

We concluded by deriving a max-min multi-letter upper bound on the mismatch capacity $C_q(W)$ which bears some resemblance to $C_q^{(\infty)}(W)$. While the upper bound does not match the known lower bounds, the proof technique may shed some light on the mismatch problem. 

\section{Acknowledgement}\label{sc: ack}
We would like to thank the anonymous reviewers for their invaluable comments and suggestions, which led to a significant simplification of the proofs and contributed to the conciseness thereof.  
\appendix

\subsection{Proof of Lemma \ref{cl: Claim 1}}\label{ap: Proof of Claim 1}

Recall the definitions of $n_x,A_x$ (\ref{eq: Ax dfn}) and recall (\ref{eq: Ax dfn 2}).
As mentioned above, $n_x$ is a deterministic quantity since all $x^n$'s lie in the same type-class. 
This implies that $q_n(X^n,Y^n)$ is the average of $|\calX|$ independent random variables $A_x,x\in\calX$, each of which is the average of at least $n\cdot \min_{x\in\calX}\hat{P}_n(x)$ i.i.d, random variables. By assumption, $P(x)>0,\forall x\in\calX$, which implies that there exists $\zeta>0$ such that $A_x$ is a average of at least $n\zeta$ i.i.d. random variables, and since the exists at least one symbol $x^*$ which satisfies $\max_{y,y'\in\calY }W(y|x^*)q(x^*,y)-W(y'|x^* ) q(x^*,y')\neq 0$, at least one of the random variables $A_x,x\in\calX$ is the average of $n\zeta$ i.i.d. random variables having a strictly positive variance. 
 Let $\calX'$ be the set of $x$'s such that the variance of $A_x$ is positive. 
We have for an arbitrary $x^n\in\calC_n$, 
\begin{flalign}\label{eq: chain of CLT}
&\Pr\left\{ q_n(X^n,Y^n)<\mathbb{E}_{\hat{P}_n\times W}(q(X,Y))-kn^{-1/2}\right\}\nonumber\\
\stackrel{(a)}{=}&\Pr\left\{\left. q_n(X^n,Y^n)<\mathbb{E}_{\hat{P}_n\times W}(q(X,Y))-kn^{-1/2}\right|X^n=x^n\right\}\nonumber\\
\stackrel{(b)}{\geq} &\Pr\left\{ \left. 
\bigcap_{x\in\calX'}
\left\{\frac{1}{n_x}\sum_{i=1}^n1\{X_i=x\} q(x,Y_i) <\mathbb{E}_{W(\cdot|x)}(q(x,Y))-kn^{-1/2}\right\}\right|X^n=x^n\right\}\nonumber\\
= &\underset{x\in\calX'}{\prod}\Pr\left\{ \left.\frac{1}{n_x}\sum_{i=1}^n1\{X_i=x\} q(x,Y_i) <\mathbb{E}_{W(\cdot|x)} (q(x,Y))-kn^{-1/2}\right|X^n=x^n\right\}. 
\end{flalign}
where $(a)$ follows since $\calC_n$ is a constant composition, and $(b)$ follows since if for all $x$, $ a_x\leq  b_x$
then for any p.m.f.\ $\mu\in\calP(\calX)$,  $\frac{1}{|\calX|} \sum_{x=1}^{|\calX|} a_x\leq \frac{1}{|\calX|}\sum_{x=1}^{|\calX|} b_x$,

Now, given $X^n=x^n$, $A_x$ is the average of $n_x$ positive-variance i.i.d.\ random variables with expectation $\mathbb{E}_{W(\cdot|x)} (q(x,Y))$.

For $\ell$ i.i.d.\ random variables $Z_i$ whose expectation and variance are $\mu$ and $\sigma^2$, respectively, we have from the CLT 
\begin{flalign}
&\lim_{\ell\rightarrow\infty}\Pr\left\{\frac{1}{\ell}\sum_{i=1}^{\ell} Z_i\leq \mu- k(\ell\cdot c)^{-1/2} \right\}\nonumber\\
=&\lim_{\ell\rightarrow\infty} \Pr\left\{\frac{1}{\sqrt{\ell}\sigma}\sum_{i=1}^{\ell} ( Z_i-\mu)\leq -kc^{-1/2}/\sigma \right\}\nonumber\\
=& \Phi(-kc^{-1/2}/\sigma)>0,
\end{flalign}
where $\Phi(\cdot)$ is the cumulative distribution function of the standard Gaussian distribution.

Applying this to (\ref{eq: chain of CLT}) we obtain 
\begin{flalign}\label{eq: chain of CLT 2}
&\Pr\left\{ q_n(X^n,Y^n)<\mathbb{E}_{\hat{P}_n\times W}(q(X,Y))-kn^{-1/2}|X^n=x^n\right\}\nonumber\\
\geq  & \left(\Phi(-k\zeta^{-1/2}/\sigma) \right)^{|\calX'|}>0,
\end{flalign}
where $\sigma^2$ is the positive variance mentioned above. 

Therefore, 
\begin{flalign}
\lim_{n\rightarrow\infty} \Pr\left\{ q_n(X^n,Y^n)<\mathbb{E}_{\hat{P}_n\times W}(q(X,Y))-kn^{-1/2}\right\}> 0, 
\end{flalign}
which concludes the proof of Lemma \ref{cl: Claim 1}.

\subsection{Proof of Lemma \ref{cl: Claim 2}}\label{ap: Proof of Claim 2}

First, note that 
\begin{flalign}
&P_{margin}(W^{(n)},\calC_n,q_n,\epsilon)\nonumber\\
= &\Pr\left\{\max_{ x^n\in\calC_n:\; x^n\neq X^n} q_n(x^n, Y^n) \geq q_n(X^n,Y^n)-\epsilon\right\}. 
\end{flalign}
Next, define the events $\calA_1=\left\{\max_{ x^n\in\calC_n:\; x^n\neq X^n}q_n(x^n, Y^n) \geq\tau \right\}$ and $\calA_2=\left\{q_n(X^n, Y^n) <\tau \right\}$, whose union is the error event of the threshold decoder; i.e., $P_{thresh}(W^{(n)},\calC_n,(q_n,\tau))=\Pr\left\{ \calA_1\cup \calA_2 \right\}$. 
We have, 
\begin{flalign}\label{eq: claim 2 result again}
&\Pr\left\{\max_{ x^n\in\calC_n:\;  x^n\neq X^n,}q_n(x^n, Y^n) \geq q_n(X^n,Y^n)-\epsilon\right\}\nonumber\\
=&\Pr\left\{\left\{\max_{ x^n\in\calC_n:\;  x^n\neq X^n} q_n(x^n, Y^n) \geq q_n(X^n,Y^n)-\epsilon\right\}\cap (\calA_1\cup \calA_2)\right\}\nonumber\\
&+\Pr\left\{\left\{\max_{ x^n\in\calC_n:\;  x^n\neq X^n} q_n(x^n, Y^n) \geq q_n(X^n,Y^n)-\epsilon\right\}\cap  (\calA_1\cup \calA_2)^c\right\}\nonumber\\
\leq &\Pr\left\{ \calA_1\cup \calA_2 \right\}\nonumber\\
&+\Pr\left\{\left\{\max_{x^n\in\calC_n:\;  x^n\neq X^n}q_n(x^n, Y^n) \geq q_n(X^n,Y^n)-\epsilon\right\}\cap  \calA_1^c\right\}\nonumber\\
\leq &P_{thresh}(W^{(n)},\calC_n,(q_n,\tau))+ \Pr\left\{ q_n(X^n,Y^n) \leq \tau+\epsilon\right\},
\end{flalign}
where the last step follows since
\begin{flalign}
&\left\{\exists x^n\in\calC_n:\;  x^n\neq X^n, q_n(x^n, Y^n) \geq q_n(X^n,Y^n)-\epsilon\right\}\cap  \calA_1^c\nonumber\\
=& \left\{\max_{ x^n\in\calC_n:\;  x^n\neq X^n}q_n(x^n, Y^n) \geq q_n(X^n,Y^n)-\epsilon\right\}\cap  
\left\{\max_{x^n\in\calC_n:\; x^n\neq X^n}q_n(x^n, Y^n) <\tau \right\}  \nonumber\\
\subseteq &\left\{ q_n(X^n,Y^n) \leq \tau+\epsilon\right\},
\end{flalign}
proving Lemma \ref{cl: Claim 2}.


\end{document}